\documentclass[a4paper,10pt]{amsart}
\usepackage{amsmath,amsfonts,amsthm}
\usepackage[latin1]{inputenc}
\usepackage[round]{natbib}

\newtheorem{theorem}{Theorem}[section]
\newtheorem{corollary}[theorem]{Corollary}
\newtheorem{lemma}{Lemma}

\theoremstyle{definition}
\newtheorem{definition}[theorem]{Definition}
\newtheorem{assumption}{Assumption}
\theoremstyle{remark}
\newtheorem{remark}[theorem]{Remark}
\newcommand{\ee}{\mathbb E}
\newcommand{\qq}{\mathbb Q}
\newcommand{\rr}{\mathbb R}
\newcommand{\pp}{\mathbb P}

\newcommand{\dd}{\text d}

\newcommand{\funcmax}[1]{\overline{#1}}

\begin{document}
\title[Hedging Errors Induced by Adaptive Discrete Trading]{Hedging Errors Induced by Discrete Trading Under an Adaptive Trading Strategy}
\author[M. Brod\'en]{Mats Brod\'en}
\address{Centre for Mathematical Sciences\\ Lund University\\ Box 118\\ 221 00 Lund, Sweden}
\email{matsb@maths.lth.se}
\author[M. Wiktorsson]{Magnus Wiktorsson}
\address{Centre for Mathematical Sciences\\ Lund University\\ Box 118\\ 221 00 Lund, Sweden}
\email{magnusw@maths.lth.se}
\begin{abstract}
Discrete time hedging in a complete diffusion market is considered. The hedge portfolio is rebalanced when the absolute difference between delta of the hedge portfolio and the derivative contract reaches a threshold level. The rate of convergence of the expected squared hedging error as the threshold level approaches zero is analyzed. The results hinge to a great extent on a theorem stating that the difference between the hedge ratios normalized by the threshold level tends to a triangular distribution as the threshold level tends to zero.
\end{abstract}
\maketitle

\section{Introduction}
In a complete market setting contingent claims can be perfectly replicated by
trading in the underlying or in some other derivative. In general if the
market is modelled by a continuous time process then the hedge portfolio has to be rebalanced at every time instant. Continuous trading is in practice impossible; a limitation giving rise to a hedging error.

In this work a discretely rebalanced hedge portfolio following the well known delta hedging strategy is considered. The hedge portfolio is rebalanced when the delta of the hedge portfolio and the derivative contract differ by some amount, here denoted by $\eta$. We investigate the hedging error, denoted by $\mathcal{R}$, as we let $\eta$ approach $0$, and show that $\ee[\mathcal{R}^2]/\eta^2$ converges to a nondegenerate limit. The limit is calculated explicitly using a result for approximations of Wiener driven SDEs,  which states that the normalized difference between the continuously evolving hedge process and the piecewise constant approximating process as we let $\eta$ approach zero is equal in distribution to a random variable which is triangularly distributed. Since the considered trading rule will give rise to a random number of stochastic times where the portfolio is rebalanced, the expected number of rebalancing times with respect to $\eta$ is also investigated.

We do not claim optimality of the investigated hedging and rebalancing strategy. Instead we justify our investigation with the resemblance between the considered strategy and how things are done in practical situations. Clearly, in practice the decision process regarding when to rebalance the hedge portfolio may be somewhat more complex, involving other risk measures such as for example vega of the derivative.

Previous work on the hedging error induced by discrete trading has mainly dealt with strategies where the hedge portfolio has been rebalanced at $n$ fixed deterministic points in time. Results often concern the order of convergence of the hedging error with respect to the number of rebalancing points $n$. One of the first contributions in this direction was done by \citet{Zhang:1999}, where discrete time hedging on an equidistant time grid was considered. For European call and put options the order of $L^2$-convergence to $0$, of the hedging error, as $n$ approaches infinity was found to be $n^{-1/2}$. In \citet{Gobet_Temam:2001} options with more irregular payoff-functions was considered. In particular, for a digital option the rate of convergence was found to be $n^{-1/4}$. In \citet{Geiss:2002} the rate of convergence using non-equidistant time nets was studied. It was found that a rate of convergence of $n^{-1/2}$ could be achieved also for options with more irregular payoff-functions, such as the digital option. \citet{Hayashi_Mykland:2005} investigates hedging errors due to discreteness of trading times or observation times of the market process. They derive the limiting distribution of the hedging error utilizing a weak convergence approach. In \citet{Tankov_Voltchkova:2009} a weak convergence approach is used to analyse the asymptotic behaviour of hedging errors in models with jumps. In \citet{Geiss_Geiss:2006} the rate of convergence for a stochastic time net is studied.

In the next section our setting and some preliminaries regarding pricing and hedging is presented. Section \ref{sec:results} contains our main results. The first result, Theorem \ref{prop:gen_diff}, is a general limit theorem for approximations of It\^{o} processes, this result is then applied in the proof of Theorem \ref{prop:main_1} which regards the convergence of the hedging error to $0$ as $\eta$ approaches $0$. Theorem \ref{prop:main_2}, which deals with the expected number of rebalancings, is followed by a corollary to Theorem \ref{prop:main_1} and \ref{prop:main_2}. We finish of with some conclusions in Section~\ref{sec:discussion} followed by an appendix containing a technical lemma.

\section{Setting and preliminaries} \label{sec:setting}
Given a filtered probability space, $\{\Omega,$ $ \mathcal{F},$ $ \{ \mathcal{F}_t \}_{t \geq 0},$ $ \pp  \}$, we let the risky asset be modeled by the stochastic process $S$. Let $\qq$ be an equivalent Martingale measure such that the discounted price process is a Martingale under this measure and let the price process, under this measure, be defined by the dynamics
\begin{align}
\dd S_t &= r S_t \dd t + \sigma(t,S_t) S_t \dd W_t, \quad S_0=s_0 , 
\label{eqn:market_dynamics}
\end{align}
where $\{ W_t \}_{t \geq 0}$ is a one dimensional Wiener process. The risk free asset, $B$, is assumed to follow the dynamics
\begin{align*}
\dd B_t = r B_t \dd t  .
\end{align*}

Let $Y$ denote the logarithm of the process $S$ and let $\tilde{\sigma}$ denote the diffusion coefficient of the process $Y$, i.e. $Y_t=\log(S_t)$ and $\tilde{\sigma}(t,y)=\sigma(t,e^y)$. Furthermore let $p_S$  denote the transition density of the process $S$ and let $p_Y$ denote the transition density of the process $Y$. Below we will introduce a set of assumptions that guarantees that the transition densities are smooth enough, which translates to smoothness properties of the pricing function and its derivatives in this market.

\begin{assumption}  \label{Assump:sigma} ${}$
\begin{list}{}{\setlength{\leftmargin}{9mm} \setlength{\labelwidth}{5mm}}
\item[H1.]
\begin{enumerate}
\item There is a positive constant $\sigma_0$ such that $\tilde{\sigma}(t,y) \geq \sigma_0$ for all $(t,y) \in [0,T] \times \mathbb{R}$.
\item The function $\tilde{\sigma}$ is bounded and uniformly Lipschitz continuous in $(t,y)$ in compact subsets of $[0,T] \times \mathbb{R}$ and uniformly H\"{o}lder continuous in $y$.
\end{enumerate}
\item[H2.] The functions $(\partial/\partial y)\tilde{\sigma}(t,y)$ and $(\partial^2/\partial y^2)\tilde{\sigma}(t,y)$ are bounded continuous functions in $[0,T] \times \mathbb{R}$ and are H\"{o}lder continuous in $y$ uniformly with respect to $(t,y)$ in compact subsets of $[0,T] \times \mathbb{R}$.
\item[H3.] The function $(\partial^3/\partial y^3)\tilde{\sigma}(t,y)$ is bounded.
\end{list}
\end{assumption}
Under assumption H1 the transition probability function of $Y$ has a density, i.e. $p_Y$ \citep[see][Theorem 6.5.4]{Friedman:2006}. The densities $p_S$ and $p_Y$ are related through
\begin{equation*}
p_S(t,s,t',s') = \frac{p_Y(t,\log(s),t',\log(s'))}{s'} \, ,
\end{equation*}
and thus under assumption H1 also $S$ has a density.

Given a payoff function, $\Phi$, and a maturity date, $T$, we will let $u$ denote the price of the corresponding contract. The price is given by the well known risk neutral pricing formula (see e.g. \citet[][]{Bjork:2009})
\begin{equation}
u(t,s) = e^{-r(T-t)} \ee_{t,s}[\Phi(S_T)] \, ,
\label{eqn:price_u}
\end{equation}
where $\ee_{t,s}$ denotes the expectation under the measure $\qq$ given that $S_t=s$. At times we will suppress this dependence on the expectation and simply write $\ee$. In this paper we will only consider the European call option, and thus throughout assume that the payoff function is given by $\Phi(x)=(x-K)^+$ for some strike price $K>0$.

It is well known that in order to replicate a contract, $u$, the hedger of the
contract should at every time instant hold the amount $(\partial/\partial s)u(t,S_t)$ of the risky asset (see e.g. \citet[][]{Bjork:2009}). This is what is commonly referred to as delta hedging. However, in order to achieve this the hedge portfolio must be continuously rebalanced. In practice this is clearly impossible. An alternative strategy is to rebalance at some discrete points in time, $\{t_i\}_{i=0,1,\ldots}$. Such a strategy will give rise to a hedging error, $\mathcal{R}$, which is defined as the difference between the value of the contract and the hedge portfolio at expiration. Let $D(t,S_t)$ denote the difference in hedge ratios at time $t$
\begin{equation}
D(t,S_t)= \frac{\partial u}{\partial s}(t,S_t) - \frac{\partial u}{\partial s}(\varphi_t,S_{\varphi_t}) ,
\label{eqn:D_fct}
\end{equation}
where $\varphi_t = \sup\{t_i \, : \, t_i \leq t \}$. The discounted hedging error at time $T$ is then given by
\begin{equation}
\mathcal{R} = \int_0^T D(t,S_t) \text{d} \tilde{S}_t ,
\label{eqn:hedging_error}
\end{equation}
where $\tilde{S}_t$ denotes the discounted prices process $e^{-rt} S_t$.

We will consider a trading strategy where the hedge portfolio is rebalanced when the hedge ratios $(\partial/\partial s)u(t,S_t)$ and $(\partial/\partial s)u(\varphi_t,S(\varphi_t))$ differ by some amount, here denoted by $\eta$. We will let $t^{\eta}_i$, defined in Definition \ref{Def:hedgehittimes} below, denote the points in time where the hedge portfolio is rebalanced for some $\eta$.

\begin{definition}
\label{Def:hedgehittimes}
Define $t^\eta_0=t_0$ and then recursively define a sequence of stopping times
\begin{equation*}
t^{\eta}_i = \inf \left\{ t > t^{\eta}_{i-1} : \left| \frac{\partial u}{\partial s}(t,S_t)-\frac{\partial u}{\partial s}(t^{\eta}_{i-1},S_{t^{\eta}_{i-1}}) \right| =\eta \right\} ,
\end{equation*}
for $i=1,2,\ldots$, and let $\varphi^\eta_t=\sup \{t^\eta_i:t^\eta_i \leq
t\}$.
\end{definition}
The hedging error induced by the hedging strategy defined by Definition~\ref{Def:hedgehittimes} will be denoted by $\mathcal{R}_{a}(\eta)$. Furthermore, in this setting we will let $N^\eta_t$ denote the number of rebalancing points up to time $t$, i.e. $N^\eta_t=\sup\{ i:t_i^{\eta} \leq t \}$, and let $H^\eta_t$ denote the expected number of rebalancing points up to time $t$, thus $H^\eta_t=\ee[N^\eta_t]$.

In this paper we will study the second moment of the hedging error, $\ee[\mathcal{R}_{a}^2(\eta)]$, as we let the expected number of rebalancing points, $H^\eta_t$, approach infinity, i.e. as $\eta \downarrow 0$.

Hedging strategies on an equidistant time net was studied in \citet{Zhang:1999}. For the case of European put and call options the rate of convergence is $1/\sqrt{n}$ and
\begin{equation}
\lim_{n \rightarrow \infty} n \ee[\mathcal{R}_{f}^2(n)]=\frac{T}{2} \ee \left[ \int_0^T e^{-2rt} S^4_t \sigma^4(t,S_t) \left( \frac{\partial^2 u}{\partial s^2}(t,S_t) \right)^2 \text{d}t \right] \, ,
\label{eq:conv_eqdist}
\end{equation}
where $\mathcal{R}_{f}(n)$ denotes the hedging error using the equidistant grid with $n$ number of rebalancing points.

\section{Results} \label{sec:results}

\subsection{A limit result for approximations of Wiener driven SDEs}
The proof of our main result regarding adaptive hedging, that is Theorem~\ref{prop:main_1}, hinges on the following theorem stating that the normalized
difference $D(t,S_t)/\eta$ converges in distribution to a triangularly
distributed random variable as $\eta$ approaches $0$. As will be seen this
result holds in a more general setting. If we let $X$ be a Wiener driven SDE
satisfying some conditions then the normalized difference
$(X_t-X_{\varphi^\eta_t})/\eta$, where $\varphi^\eta_t$ is defined as in
Definition~\ref{Def:hittimes} below, converges in distribution to a triangularly
distributed random variable as $\eta$ approaches $0$.

Let $X$ be defined by
\begin{equation}
\dd X_t=a(t,X_t)\dd t+b(t,X_t)\dd W_t, \quad X_{t_0}=x_0,
\label{Eq:SDEX}
\end{equation}
We can assume that $b$ is non-negative since the distributional properties of $X$ will not change if we replace $b$ with $|b|$. 
\begin{definition}
\label{Def:hittimes}
Define $t^\eta_0=t_0$ and then recursively define a sequence of stopping times
\begin{equation*}
t^\eta_i=\inf \{t>t^\eta_{i-1}: |X_t-X_{t^\eta_{i-1}}|=\eta\},
\end{equation*}
for $i=1,2,\ldots$, and let $\varphi^\eta_t=\sup \{t^\eta_i:t^\eta_i \leq
t\}$.
\end{definition}
\begin{assumption}\label{Assump:X}${}$
\begin{list}{}{\setlength{\leftmargin}{9mm} \setlength{\labelwidth}{5mm}}
\item[A1.]
Equation \eqref{Eq:SDEX} has a unique globally existing solution
for $t\in[t_0,t_f]$.
\item[A2.]
$\ee[|h(t,X_t)|]<\infty$  for $t\in[t_0,t_f]$ where $h(t,x)\in C([t_0,t_f]\times D(x_0))$ and $D(x_0)$ are the set of points the solution can visit starting at $x_0$.
\item[A3.]
For each $\epsilon>0$ there exists $\delta_{\epsilon,h}>0$ such that $\ee[|h(t,X_t)|I_{A_{\delta_{\epsilon,h}}}(X_t)]<\epsilon$ for $t\in (t_0,t_f)$, where $\underline{b}(x)=\inf_{t\in[t_0,t_f]}b(t,x)$ and $A_\delta=\{x\in
 D(x_0):\underline{b}(x)<\delta\}$.  
\item[A4.]
The process $X$ has a bounded continuous density, p(t,x), for all $(t,x) \in (t_0,t_f) \times A_{\delta_{\epsilon,h}}^c$.
\end{list}
\end{assumption}

\begin{remark}
The reason for the somewhat technical assumption A3 is to allow for processes
which can enter a region where $b$ is arbitrary close to zero. If the function
$h$ is such that this does not contribute to the expectation things will work
out anyway. Of course if $b$ is bounded from below by a positive constant
we easily see that A3 is satisfied.
\end{remark}

\begin{theorem}\label{prop:gen_diff}
 Let $f\in C_b([-1,1])$, let $X$ and $h$ satisfy Assumption \ref{Assump:X} and
let $\varphi^\eta_t$ be given by Definition \ref{Def:hittimes}, then
\begin{align*}
\lim_{\eta\downarrow0}\ee\left[f\left(\frac{1}{\eta}(X_t-X_{\varphi^\eta_t})\right)h(t,X_t)\right]&= \lim_{\eta\downarrow0}\ee\left[f\left(\frac{1}{\eta}(X_t-X_{\varphi^\eta_t})\right)\right]\ee\left[h(t,X_t)\right]\\
&=\int_{-1}^1f(z)(1-|z|)\dd z\ee\left[h(t,X_t)\right],
\end{align*}
for $t\in(t_0,t_f)$.
\end{theorem}

\begin{proof}
We can without loss of generality assume that $x_0=0$. Since for each arbitrary but fixed $x_0$ we can define $\tilde{X}_t=X_t-x_0,~\tilde{h}(t,x)=h(t,x+x_0),~\tilde{a}(t,x)=a(t,x+x_0),~\tilde{b}(t,x)=b(t,x+x_0),~\tilde{p}(t,x)=p(t,x+x_0)$. Now
\begin{equation*}
\dd \tilde{X}_t=\tilde{a}(t,\tilde{X}_t)\dd t+\tilde{b}(t,\tilde{X}_t)\dd W_t,~\tilde{X}_0=0.
\end{equation*}
However to keep the notation simple we skip these tildes from now on.
  
We have two cases two distinguish between, first $\{\varphi^\eta_t=0\}=\{t^\eta_1>t\}$ and second  $\{\varphi^\eta_t>0\}=\{t^\eta_1 \leq t\}$. The probability of being in the first case goes to zero as $\eta$ tends to zero. Since $f$ is bounded and $h$ is locally bounded by continuity it is straightforward to establish that
\begin{equation*} 
\ee\left[f(X_t/\eta)h(t,X_t)I(t^\eta_1>t)\right]\leq C\sup_{|x|\leq \eta}|h(t,x)|P(t^\eta_1>t),
\end{equation*}
which tends to zero for any fixed $x_0$ as $\eta$ tends downwards to zero. So the contribution to the limit comes from the second case.

On the set $\{t^\eta_1\leq t\}$ we have that $D_{\eta}(t)=(X_t-X_{\varphi^\eta_t})/\eta$ can take the values  $X_t/\eta-\lfloor X_t/\eta\rfloor$ and $X_t/\eta-\lfloor X_t/\eta\rfloor-1$ we denote the probability for the first value conditioned on $X_t=x$ as $\bar{p}(t,x,\eta\lfloor x/\eta\rfloor)$, where $\lfloor y \rfloor$ means rounding the real number $y$ downwards to the nearest integer. We postpone the calculation of this probability until later.

Since $\ee[|f(D_\eta(t))h(t,X_t)|]<\infty$ we can for each $\epsilon>0$ find a compact set (closed interval) $K_{\epsilon}$ such that  $\ee[|f(D_\eta(t))h(t,X_t)|I_{K_\epsilon^c}(X_t)]<\epsilon/4$. Further by assumption A3 we have that we can find $\delta_{\epsilon/4,h}$ such that
\begin{equation*}
\ee[|f(D_\eta(t))h(t,X_t)|I_{A_{\delta_{\epsilon/4,h}}}(X_t)]<\epsilon/4.
\end{equation*}
Now let $B_\epsilon=A_{\delta_{\epsilon/4,h}}^c\cap K_{\epsilon}$. We then have that
\begin{equation*}
|\ee[f(D_\eta(t))h(t,X_t)]-\ee[f(D_\eta(t))h(t,X_t)I_{B_\epsilon}(X_t)]|<\epsilon/2.
\end{equation*}
Since $f$ is bounded and $h$ is uniformly continuous on compacts there exists $\eta_\epsilon$ such that for $\eta<\eta_\epsilon$
\begin{equation*}
|\ee[f(D_\eta(t))h(t,X_t)I_{B_\epsilon}(X_t)]-\ee[f(D_\eta(t))h(t,\eta\lfloor X_t/\eta\rfloor)I_{B_\epsilon}(X_t)]|<\epsilon/4.
\end{equation*}
The last expectation may now be written as an integral with respect to the density $p(t,x)$
\begin{multline*}
\int_{\rr}f(x/\eta-\lfloor x/\eta\rfloor)\bar{p}(t,x,\eta\lfloor x/\eta\rfloor)){h}(t,\eta\lfloor x/\eta\rfloor)p(t,x))I_{B_\epsilon}(x)\dd x \\
+\int_{\rr}f(x/\eta-\lfloor x/\eta\rfloor-1)(1-\bar{p}(t,x,\eta\lfloor x/\eta\rfloor){h}(t,\eta\lfloor x/\eta\rfloor)p(t,x))I_{B_\epsilon}(x)\dd x.
\end{multline*}
By decreasing $\eta$ even further, $\eta< \eta'_\epsilon$ say, these integrals can be approximated by an error of at most $\epsilon/4$ by the integrals
\begin{align*}
I_1+I_2&=\int_{\rr}f(x/\eta-\lfloor x/\eta\rfloor)\bar{p}(t,x,\eta\lfloor x/\eta\rfloor)) \\
& \qquad \times {h}(t,\eta\lfloor x/\eta\rfloor)p(t,\eta\lfloor x/\eta\rfloor)I_{B_\epsilon}(\eta\lfloor x/\eta\rfloor)\text{d}x\\
& +\int_{\rr}f(x/\eta-\lfloor x/\eta\rfloor-1)(1-\bar{p}(t,x,\eta\lfloor x/\eta\rfloor) \\
& \qquad \times {h}(t,\eta\lfloor x/\eta\rfloor)p(t,\eta\lfloor x/\eta\rfloor)I_{B_\epsilon}(\eta\lfloor x/\eta\rfloor)\dd x.
\end{align*}
We then finally arrive at
\begin{equation*}
|\ee[f(D_\eta(t))h(t,X_t)]-I_1-I_2|<\epsilon,~\text{for}~\eta< \eta'_\epsilon.
\end{equation*}

We can now treat the integrals $I_1$ and $I_2$ separately. The integral $I_1$ is given by
\begin{equation*}
I_1=\sum_{k=-\infty}^\infty \int_{k\eta}^{(k+1)\eta}f(x/\eta-k)\bar{p}(t,x,k\eta){h}(t,k\eta)p(t,k\eta)I_{B_\epsilon}(k\eta)\dd x.
\end{equation*}
Making the change of variables $x=z\eta+k\eta$ we obtain
\begin{align*}
I_1&=\sum_{k=-\infty}^\infty \eta\int_{0}^{1}f(z)\bar{p}(t,z\eta+k\eta,k\eta){h}(t,k\eta)p(t,k\eta)I_{B_\epsilon}(k\eta)\dd z\\
&=\sum_{k=-\infty}^\infty {h}(t,k\eta)p(t,k\eta)I_{B_\epsilon}(k\eta)\eta\int_{0}^{1}f(z)\bar{p}(t,z\eta+k\eta,k\eta)\dd z.
\end{align*}
Similar calculations for $I_2$ leads to
\begin{equation*}
I_2=\sum_{k=-\infty}^\infty {h}(t,k\eta)p(t,k\eta)I_{B_\epsilon}(k\eta)\eta\int_{0}^{1}f(z-1)(1-\bar{p}(t,z\eta+k\eta,k\eta))\dd z.
\end{equation*}
We are now ready to deal with the probability $\bar{p}$. Under fairly mild condition on the density $p(t,x)$, the distributional  derivative $(\partial/\partial x) (p(t,x)b^2(t,x))$ should be a locally integrable function \citep[see][]{Millet_Nualart_Sanz:1989}, there exists a time-reversal for the diffusion $X$. We can now calculate the probability $\bar{p}$ as the hitting probability of a time reversed diffusion starting at $x$ who should hit $\eta\lfloor x/\eta\rfloor$ before $\eta\lfloor x/\eta\rfloor+\eta$. 

Let $\bar{X}_s=X_{t-s}$ be the time reversed diffusion related to $X$. From \citet{Millet_Nualart_Sanz:1989} we have that the dynamics of $\bar{X}$ for $0\leq s<t$ is
\begin{align*}
\dd \bar{X}_s&=\bar{a}(s,\bar{X}_s)\dd s+\bar{b}(s,\bar{X}_s)\dd W_s, \quad \bar{X}_0 =X_t,
\end{align*} 
where
\begin{align*}
\bar{a}(s,x)&=-a(t-s,x) \\
& \qquad +\frac{1}{p(t-s,x)}\frac{\partial}{\partial x}(b(t-s,x)^2p(t-s,x))I(p(t-s,x)\neq 0),\\
\bar{b}(s,x)&=b(t-s,x).
\end{align*} 
We now look at the probability for the time reversed diffusion of hitting the
point $c$ before the point $c+\eta$ starting at $x$, for $c\leq x \leq c+\eta$, and denote this $g(t,x)$. Using the results of \citet[][Chapter 9-11]{Oksendal:1998} we obtain that
$g$ satisfies the following partial differential equation (PDE)
\begin{align*}
\frac{\partial}{\partial t}g(t,x)+\bar{a}(s,x)\frac{\partial}{\partial x}g(t,x)+\frac{\bar{b}(s,x)^2}{2}\frac{\partial^2}{\partial x^2}g(t,x) & =0,\\
g(t,c)=1, \quad g(t,c+\eta) & =0.
\end{align*} 
We now look at $\bar{g}(t,z)=g(t,c+z\eta)$ for $0\leq z \leq 1$. By using the PDE for $g$ we now find that the corresponding PDE for $\bar{g}$ is given by
\begin{align*}
\frac{\partial}{\partial t}\bar{g}(t,z)+\bar{a}(s,c+z\eta)\frac{\partial}{\partial z}
\bar{g}(t,z)/\eta+\frac{\bar{b}(s,c+z\eta)^2}{2}\frac{\partial^2}{\partial z^2}\bar{g}(t,z)/\eta^2 & =0,\\
\bar{g}(t,0)=1, \quad \bar{g}(t,1) & =0,
\end{align*} 
which is equivalent to
\begin{align*}
\eta^2\frac{\partial}{\partial t}\bar{g}(t,z)+\eta\bar{a}(s,c+z\eta)\frac{\partial}{\partial z}\bar{g}(t,z)+\frac{\bar{b}(s,c+z\eta)^2}{2}\frac{\partial^2}{\partial z^2}\bar{g}(t,z) & =0,\\
\bar{g}(t,0)=1, \quad \bar{g}(t,1) & =0.
\end{align*} 
Now letting $\eta$ tend to zero we obtain the limiting PDE, where
we let $\hat{g}$ be the limit of $\bar{g}$
\begin{align*}
\frac{\bar{b}(s,c)^2}{2}\frac{\partial^2}{\partial z^2}\hat{g}(t,z)=0,\\
\hat{g}(t,0)=1, \quad \hat{g}(t,1)=0,
\end{align*} 
which has the solution $\hat{g}(t,z)=1-z$ under the assumption that $\bar{b}$ is always bounded away from zero (which is assured for $c\in B_{\epsilon}$) and thus the limit of the  probability $\bar{p}$ will be $1-z$ which is the same as the hitting probability for a standard Brownian motion.

We can now apply this result to the integrals $I_1$ and $I_2$. Using 
dominated convergence we obtain
\begin{align*}
\lim_{\eta \downarrow 0}  I_1+I_2& = \ee[h(t,X(t))]\left(\int_0^1f(z)(1-z)\text{d}z+\int_0^1f(z-1)z\text{d}z\right)+\epsilon\\
&= \ee[h(t,X(t))]\int_{-1}^1f(z)(1-|z|)\text{d}z+\epsilon.
\end{align*}
Since $\epsilon$ can be made arbitrary small we have shown the desired result.
\end{proof}

\subsection{Adaptive discrete time hedging} \label{sec:dischedging}
Using (\ref{eqn:hedging_error}) and the It\^{o} isometry we get the following bound of the second moment of the hedging error
\begin{equation*}
\begin{split}
\ee [\mathcal{R}_a^2(\eta)] & = \ee \left[ \left( \int_0^{T} D(t,S_t) \dd \tilde{S}_t \right)^2 \right] \\
& = \ee \left[ \int_0^{T} D(t,S_t)^2 e^{-2 r t} S_t^2  \sigma^2(t,S_t) \dd t \right] \\
& \leq \eta^2 \ee \left[ \int_0^{T} e^{-2 r t}  S_t^2  \sigma^2(t,S_t) \dd t \right] \, .
\end{split}
\end{equation*}

Below we state one of our main results which regards the rate of convergence with respect to the parameter $\eta$.

\begin{theorem} \label{prop:main_1}
Let the payoff function be defined by $\Phi(x)=(x-K)^+$, let $\varphi_t^{\eta}$ be given by Definition \ref{Def:hedgehittimes}, and let $\tilde{\sigma}$ satisfy Assumption \ref{Assump:sigma}, then
\begin{equation}
\lim_{\eta \downarrow 0} \frac{1}{\eta^2} \ee[\mathcal{R}_a(\eta)^2]=\frac{1}{6} \ee \left[ \int_0^T e^{-r t}  S^2_t \sigma^2(t,S_t) \dd t \right] \, .
\label{eq:roc_1}
\end{equation}
\end{theorem}

\begin{proof} Below we will let $C$ and $C'$ denote bounded positive constants whose values may change from line to line. Furthermore, for notational convenience we will at times write $f_x(x,y)=(\partial/\partial x)h(x,y)$, for some function $f$, and equivalently for higher order derivatives and derivatives with respect to the other variable $y$.

\textit{Step 1.} Define the process $X$ and the function $g$ by $X_t=g(t,S_t)=(\partial/\partial s)u(t,S_t)$. Using It\^{o}'s formula the dynamics of the process $X$ is given by
\begin{multline*}
\dd X_t = \left\{ \frac{\partial g}{\partial t}(t,S_t)+ r S_t \frac{\partial g}{\partial s}(t,S_t) +\frac{1}{2}\sigma^2(t,S_t)S^2_t \frac{\partial^2 g}{\partial s^2}(t,S_t) \right\} \dd t \\
+ \sigma(t,S_t) S_t \frac{\partial g}{\partial s}(t,S_t) \dd W_t.
\end{multline*}
Since $u$ satisfies the Black-Scholes PDE we get by taking derivatives with respect to
$s$ and substituting into the drift term that
\begin{equation}
\dd X_t=-\frac{\partial}{\partial s}\left(\frac{\sigma^2(t,S_t)S^2_t}{2}\right)\frac{\partial g}{\partial s}(t,S_t) \dd t +\sigma(t,S_t)S_t \frac{\partial g}{\partial s}(t,S_t) \dd W_t.
\label{eqn:X_ito}
\end{equation}
Due to Lemma \ref{lem:bnds} (i), for each $t \in [\varepsilon,T-\varepsilon]$, the function $g(t,s)$ is continuous and strictly increasing in $s$, and thus there exists an inverse function such that $g^{-1}(t,g(t,s))=s$. Hence, \eqref{eqn:X_ito} can be written as
\begin{equation}
\dd X_t=a(t,X_t)\dd t+b(t,X_t)\dd W_t , \label{eqn:X_ito_ab}
\end{equation}
for some functions $a$ and  $b$. As in (\ref{eqn:D_fct}) we let $D(t,S_t)$ denote the difference in hedge ratios at time $t$. Using the It\^{o} isometry and the above definition of $X$ the expected squared hedging error may be written as
\begin{equation}
\begin{split}
\frac{1}{\eta^2} \ee[\mathcal{R}_a^2(\eta)] & = \frac{1}{\eta^2} \ee \left[ \left( \int_0^{T} D(t,S_t)  \dd \tilde{S}_t \right)^2 \right] \\
& = \int_0^{T} \ee \left[ \frac{1}{\eta^2}(X_t-X_{\varphi^\eta_t})^2 e^{-2  r t} S_t^2 \sigma^2(S_t)  \right] \dd t . \label{eqn:Rdiveta}
\end{split}
\end{equation}
Next we would like to apply Theorem \ref{prop:gen_diff} in order to handle the above expression as we let $\eta$ approach zero. In the setting of Theorem \ref{prop:gen_diff} we will let $h(t,x)=(g^{-1}(t,x))^2\sigma(g^{-1}(t,x))^2$ and $f(x)=x^2$. Furthermore, for some small value $\varepsilon>0$, we will divide the integral in \eqref{eqn:Rdiveta} into three parts; the first integral going from $0$ to $\varepsilon$, the second from $\varepsilon$ to $T-\varepsilon$ and the last from $T-\varepsilon$ to $T$. Since $D^2(t)/\eta^2$, $\sigma^2(t,S_t)$, $\exp\{ - 2 r t\}$ and $\ee[S_t^2]$ \citep[see][the proof of Theorem 5.1.1]{Friedman:2006} are bounded it is seen that the first and the last integral may be bounded by $\varepsilon$ times some bounded constant, and thus can be made arbitrary small by simply decreasing $\varepsilon$. Theorem \ref{prop:gen_diff} may now be applied in the region $[\varepsilon,T-\varepsilon]$. In the following steps we will check that Assumption~\ref{Assump:X} is satisfied.

\textit{Step 2.} In this step it will be shown that $X$ has a solution (i.e. assumption A1). If the coefficients of \eqref{eqn:X_ito_ab} are locally Lipschitz continuous by \citet[Theorem II.5.2]{Kunita:1984} the SDE \eqref{eqn:X_ito_ab} has a unique solution $X$ up to a possibly finite random explosion time (see also the proof of Theorem 1 in \citet[][]{Heath_Schweizer:2000}). However, since $|X|$ by Lemma \ref{lem:bnds} (i) is bounded by 1, what is left to show is local Lipschitz continuity of the coefficients of \eqref{eqn:X_ito_ab}. The derivative of $b$ with respect to $x$ is given by
\begin{align*}
\frac{\partial b}{\partial x}(t,x) & = \sigma_y(t,g^{-1}(t,x)) g^{-1}(t,x)+\sigma(t,g^{-1}(t,x)) \\
& \qquad + \frac{\sigma(t,g^{-1}(t,x)) g^{-1}(t,x) g_{ss}(t,g^{-1}(t,x))}{g_s(t,g^{-1}(t,x))},
\end{align*}
where the first two terms are bounded due to H1 and H2. To estimate the last term we will use \eqref{lem:bnd1} and \eqref{lem:bnd2} and it is seen that there is a bounded positive constant $C$ such that
\begin{align*}
\left| \frac{\sigma(t,g^{-1}(t,x)) g^{-1}(t,x) g_{ss}(t,g^{-1}(t,x))}{g_s(t,g^{-1}(t,x))}\right| \leq C e^{C (\log(g^{-1}(t,x)))^2}.
\end{align*}
Since $g^{-1}(t,x)$ is bijective in $x$ for every $t \in [\epsilon,T-\epsilon]$, and due to the limits in Lemma \ref{lem:bnds} (i), the right hand side in the above equation goes to infinity if and only if $x$ goes to $0$ or $1$. Thus for every $t \in [0,T)$, $b(t,x)$ is locally Lipschitz continuous on the interval $(0,1)$. The last thing to check is that the process $X$ never leaves the interval $(0,1)$, which is equivalent to check that the process $Y_t=\log(S_t)$ never takes the values $-\infty$ or $\infty$. However, by a straightforward application of Feller's test for explosions \citep[see][Theorem 3.5.29]{Karatzas_Shreve:1991} it is seen that the process $Y$ never explodes in a finite time, and thus $X$ never leaves $(0,1)$ in $[\epsilon,T-\epsilon]$. The local Lipschitz continuity of $a$ is shown in the same way as for $b$.

\textit{Step 3.} Shifting to the process $S$ the expectation in assumption A2 equals $\ee[\exp\{ -2 r t\}\sigma^2(t,S_t) S^2_t]$, where $\sigma$ is bounded due to H1 and $\ee[S^2_t]$ is bounded for $t \in [\epsilon,T-\epsilon]$ due to \citet[the proof of Theorem 5.1.1]{Friedman:2006}.

\textit{Step 4.} In this step we will show that assumption A3 holds. Using the first inequality in \eqref{lem:bnd1} of Lemma \ref{lem:bnds} we have that
\begin{align}
s \sigma(t,s) \frac{\partial^2 u}{\partial s^2}(t,s) & \geq C^{-1} e^{-C \log^2(s)} , \label{eqn:dgdsbnd}
\end{align}
for some bounded constant $C$. Now, consider the log process $Y$ of $S$, and let $\underline{b}^Y(y)=\inf_{t \in [\varepsilon,T-\varepsilon]} e^y \tilde{\sigma}(t,y) g_s(t,e^y)$. Define the sets
\begin{align*}
A_{\delta}^Y = \{y \in \rr : \underline{b}^Y(y) < \delta \} && \text{and} && \overline{B}_{\delta}^Y = \{ y \in \rr : C^{-1} e^{-C y^2} < \delta \}, 
\end{align*}
then $A_{\delta}^Y \subset \overline{B}_{\delta}^Y$. The expectation in assumption A3 of Theorem \ref{prop:gen_diff} now reads 
\begin{align*}
\ee \left[|h(t,X_t)| I_{A_{\delta_{\epsilon,h}}}(X_t) \right] & = \ee \left[ e^{-2 r t} \tilde{\sigma}^2(t,Y_t) e^{2Y_t}  I_{A^Y_{\delta_{\epsilon,h}}}(Y_t) \right] \\
& \leq e^{-2 t r} \funcmax{\sigma}^2 \ee \left[  e^{2Y_t}  I_{\overline{B}^Y_{\delta_{\epsilon,h}}}(Y_t)  \right],
\end{align*}
where $\funcmax{\sigma}=\sup \sigma(t,s)$. Below we will suppress the two first arguments of the transition density $p_Y$ and write $p_Y(t,y)=p_Y(0,\log(s_0),t,y)$. According to \citet[Theorem 6.4.5]{Friedman:2006} there exists a bounded constant $C$ such that 
\begin{align}
p_Y(t,y) \leq \frac{C}{\sqrt{t}} e^{-\frac{(y-\log(s_0))^2}{C t}}, \label{eqn:pYbnd}
\end{align}
for all $(t,y) \in \mathbb{R}_+ \times \mathbb{R}$. Thus, there is a bounded constant $C$ such that $p_Y(t,y)  \leq C \exp\{-(y-\log(s))^2 C^{-1}\}$ for all $(t,y) \in [\epsilon,T-\epsilon] \times \mathbb{R}$. Using this we get that
\begin{align*}
\ee \left[  e^{2Y_t}  I_{\overline{B}^Y_{\delta_{\epsilon,h}}}(Y_t) \right] & = \int_{\overline{B}^Y_{\delta_{\epsilon,h}}} e^{2 y} p_Y(t,y) \dd y  \leq C \int_{\overline{B}^Y_{\delta_{\epsilon,h}}} e^{2 y}  e^{-\frac{(y-\log(s_0))^2}{C}}  \dd y,
\end{align*}
which holds for all $t \in [\varepsilon,T-\varepsilon]$. Noting that $\overline{B}^Y_{\delta_{\epsilon,h}}=\{ y \in \rr : y^2 > -C^{-1} \log(C \delta)  \}$
and using the inequality $-x^2+cx<-x^2/2+c^2/2$ which holds for $x,c \in \rr$, we have that there are constants $C$ and $C'$ such that
\begin{align*}
\ee \left[  e^{2Y_t}  I_{\overline{B}^Y_{\delta_{\epsilon,h}}}(Y_t) \right] & \leq C' \int_{\sqrt{-C^{-1}  \log(C \delta)}}^{\infty}  e^{-\frac{y^2}{C'}} \dd y  \leq C' e^{\frac{C^{-1} \log(C \delta)}{C'}},
\end{align*}
which holds for all $t \in [\varepsilon,T-\varepsilon]$. Hence, for every $\epsilon>0$ we may chose $\delta_{\epsilon,h}$ such that $\delta_{\epsilon,h}<C^{-1} (\epsilon/C')^{C' C}$, and thus assumption A3 is satisfied.

\textit{Step 5.} In this step it will be shown that assumption A4 holds. First we will construct a set $\underline{B}_\delta^Y$ such that $(A^Y_\delta)^c \subset (\underline{B}^Y_\delta)^c$ and then show that the required properties of $p$ are satisfied in this set. Lemma \ref{lem:bnds} yields that there is a bounded constant $C$ such that
\begin{align}
\sigma(t,s) s  \frac{\partial g}{\partial s}(t,s) & \leq C e^{-\frac{\log^2(s)}{C}}, 
\label{eqn:dgdsbndup}
\end{align}
which holds for all $(t,s) \in [\varepsilon,T-\varepsilon] \times [0,\infty)$. Define the set $\underline{B}^Y_\delta$ by $\underline{B}^Y_\delta = \{ y \in \rr : C \exp\{-y^2 C^{-1}\}  < \delta \}$ then $(\underline{B}^Y_\delta)^c = \{ y \in \rr : C \exp\{-y^2 C^{-1} \} \geq \delta \}$ and $(A^Y_\delta)^c \subset (\underline{B}^Y_\delta)^c$. Furthermore, note that $\underline{B}^Y_\delta$ can be expressed as $(\underline{B}^Y_\delta)^c  = \{ y \in \rr : y^2 \leq C \log(\delta /C) \}$ and hence is bounded for every choice of $\delta>0$. The density of the process $X$ is given by
\begin{align*}
p(t,x)= \frac{p_Y(t,\log g^{-1}(t,x))}{g^{-1}(t,x) g_s(t,g^{-1}(t,x))},
\end{align*}
and thus
\begin{align*}
p(t,g(t,e^y))= \frac{p_Y(t,y)}{e^y g_s(t,e^y)}.
\end{align*}
By \eqref{eqn:pYbnd} we have that there is a bounded constant $C$ such that $|p_Y(t,y)| \leq C$ for all $(t,y) \in [\varepsilon,T-\varepsilon] \times (\underline{B}^Y_\delta)^c$. Using Lemma \ref{lem:bnds} (ii) we have that there is a bounded constant $C$ such that $1/(e^y g_s(t,e^y)) \leq C$ for all $(t,y) \in [\varepsilon,T-\varepsilon] \times (\underline{B}^Y_\delta)^c$. Consequently $p(t,g(t,e^y))$ is bounded for all $(t,y) \in [\varepsilon,T-\varepsilon] \times (\underline{B}^Y_\delta)^c$. Differentiating $p$ with respect to $x$ yields
\begin{multline*}
p_x(t,x) = \frac{\partial p_Y}{\partial y}(t,y) \bigr|_{s=\log g^{-1}(t,x)}  \frac{1}{(g^{-1}(t,x) g_s(t,g^{-1}(t,x)))^2} \\
- \frac{p_Y(t,\log g^{-1}(t,x))}{g^{-1}(t,x)^3 g_s(t,g^{-1}(t,x))^2} - \frac{p_Y(t,\log g^{-1}(t,x)) g_{ss}(t,g^{-1}(t,x))}{(g^{-1}(t,x) g_s(t,g^{-1}(t,x)))^2},
\end{multline*}
and consequently
\begin{multline*}
p_x(t,g(t,e^y)) = \frac{\partial p_Y}{\partial y}(t,y)  \frac{1}{e^{2y} g_s^2(t,e^y)} - \frac{p_Y(t,y)}{e^{3y} g_s^2(t,e^y)} - \frac{p_Y(t,y) g_{ss}(t,e^y)}{e^{2y} g_s^2(t,e^y)}.
\end{multline*}
Using \citet[][Theorem 6.4.7]{Friedman:2006}, which holds under H1 and H2, and \citet[][Theorem 6.4.5]{Friedman:2006}, which holds under H1, we have that there exists a bounded constant $C$ such that $|(\partial/\partial y)p_Y(t,y)| \leq C$ for all $(t,y) \in [\varepsilon,T-\varepsilon] \times (\underline{B}^Y_\delta)^c$ (alternatively see \citet[][Theorem 9.6.7]{Friedman:2008}). The expression $1/g_s(t,e^y)$ may be bounded using Lemma \ref{lem:bnds} (ii) and it is seen that $1/g_s(t,e^y)$ is bounded for all $(t,y) \in [\varepsilon,T-\varepsilon] \times (\underline{B}^Y_\delta)^c$. From Lemma \ref{lem:bnds} (ii) we also have that there exists a bounded constant $C$ such that $|g_{ss}(t,e^y)| \leq C$ and thus also this expression is bounded for $(t,y) \in [\varepsilon,T-\varepsilon] \times (\underline{B}^Y_\delta)^c$. Adding up, we have that $p_x(t,g(t,e^y))$ is bounded for all $(t,y) \in [\varepsilon,T-\varepsilon] \times (\underline{B}^Y_\delta)^c$. Hence, $p(t,g(t,e^y))$ is bounded and continuous for all $(t,y) \in [\varepsilon,T-\varepsilon] \times (\underline{B}^Y_\delta)^c$ for every $\delta>0$ and $\varepsilon>0$ as was to be shown in this step.

\textit{Step 6.} Applying Theorem \ref{prop:gen_diff} yields
\begin{multline*}
\lim_{\eta \downarrow 0} \int_{\varepsilon}^{T-\varepsilon} \ee \left[ \frac{1}{\eta^2}(X_t-X_{\varphi^\eta_t})^2 e^{-2 r t} S_t^2 \sigma^2(t,S_t)  \right] \dd t \\
=  \frac{1}{6} \int_{\varepsilon}^{T-\varepsilon} \ee \left[ e^{-2 r t} S_t^2 \sigma^2(t,S_t)  \right] \dd t.
\end{multline*}
Since the contribution from the integrals over $[0,\varepsilon]$ and over $[T-\varepsilon,T]$ can be made arbitrary small by simply decreasing $\varepsilon$, as described earlier, this finishes the proof.
\end{proof}

\begin{remark}
As seen in the proof of Theorem \ref{prop:main_1}, in order to apply Theorem \ref{prop:gen_diff} the inverse of the delta of the option is needed. This is the reason why we choosed to fix the payoff function of the option. For a digital option the delta is not invertible, and thus we would not be able to apply Theorem \ref{prop:gen_diff}. For power options with payoff functions of the type $(x-K)^p$, $p > 1$, invertability is not a problem. However in this case we would instead need to modify Lemma \ref{lem:bnds}.
\end{remark}

The following results deals with the rate of convergence of the expected number of rebalancings $H^\eta_T$ with respect to $\eta$ and the order of convergence of the expected squared hedging error with respect to~$H^\eta_T$.

\begin{theorem} \label{prop:main_2} 
Let the payoff function be defined by $\Phi(x)=(x-K)^+$, let $\varphi_t^{\eta}$ be given by Definition \ref{Def:hedgehittimes}, and let $\tilde{\sigma}$ satisfy Assumption \ref{Assump:sigma}, then
\begin{equation}
\lim_{\eta \downarrow 0} \eta^2 H^\eta_T = \ee \left[ \int_0^T S^2_t \sigma^2(t,S_t) \left( \frac{\partial^2 u}{\partial s^2}(t,S_t) \right)^2  \dd t \right] < \infty .
\label{eq:roc_2_b}
\end{equation}
\end{theorem}

\begin{proof}
Let as in the proof of Theorem \ref{prop:main_1} the process $X$ be defined as $X_t=g(t,S_t)=(\partial/\partial s)u(t,S_t)$. The dynamics of $X$ is given by (see \eqref{eqn:X_ito})
\begin{equation*}
\dd X_t = \bar{a}(t,S_t) \dd t + \bar{b}(t,S_t) \dd W_t,
\end{equation*}
where
\begin{align*}
\bar{a}(t,S_t) = -\frac{\partial}{\partial s}\left(\frac{\sigma^2(t,S_t)S^2_t}{2}\right)\frac{\partial g}{\partial s}(t,S_t) && \text{and} && \bar{b}(t,S_t) = \sigma(t,S_t)S_t \frac{\partial g}{\partial s}(t,S_t).
\end{align*}
Since the distance between $X_{t_{i}^{\eta}}$ and $X_{t_{i+1}^{\eta}}$ always equals $\eta$ it holds that
\begin{align}
\eta^2 N^{\eta}_T = \sum_{i=0}^{N_T^\eta-1} (X_{t_{i+1}^{\eta}}-X_{t_{i}^{\eta}})^2. \label{eqn:id1}
\end{align}
Letting $t_{N_T^{\eta}+1}=T$, using \eqref{eqn:id1} and that
\begin{align*}
[X]_T=X_T^2-X_0^2-2 \int_0^T X_s \text{d}X_s
\end{align*}
we get that
\begin{align*}
& |\ee[\eta^2 N_T^\eta-[X]_T]|  \\
& = \Biggl|\ee \Biggl[\sum_{i=0}^{N_T^\eta} (X_{t_{i+1}^{\eta}}-X_{t_{i}^{\eta}})^2-(X_{T}-X_{t_{N_T^\eta}^{\eta}})^2 -X_T^2+X_0^T + 2 \int_0^T X_s \text{d}X_s  \Biggr] \Biggr| \\
& = \Biggl|\ee \Biggl[-2\sum_{i=0}^{N_T^\eta} X_{t_{i}^{\eta}}(X_{t_{i+1}^{\eta}}-X_{t_{i}^{\eta}})-(X_{T}-X_{t_{N_T^\eta}^{\eta}})^2 + 2\sum_{i=0}^{N_T^\eta} \int_{t_{i}^{\eta}}^{t_{i+1}} X_s \text{d}X_s  \Biggr] \Biggr|.
\end{align*}
Using that $|X_{T}-X_{t_{N_T^\eta}^{\eta}}|\leq \eta$ together with the above equation we get the following inequality
\begin{align*}
|\ee[\eta^2 N_T^\eta-[X]_T]| & \leq \Biggl|\ee \Biggl[2 \int_{0}^{T} (X_s-X_{\varphi_s^{\eta}}) \text{d}X_s\Biggr] \Biggr| + \eta^2.
\end{align*}
The first term on the right hand side in the equation above may be bounded as
\begin{align*}
& \Biggl|\ee \Biggl[2 \int_{0}^{T} (X_s-X_{\varphi_s^{\eta}}) \text{d}X_s\Biggr] \Biggr| \\
&  = 2 \Biggl|\ee \Biggl[\int_{0}^{T} (X_s-X_{\varphi_s^{\eta}}) \bar{a}(s,S_s) \text{d}s  + \int_{0}^{T} (X_s-X_{\varphi_s^{\eta}}) \bar{b}(s,S_s) \text{d}W_s \Biggr] \Biggr| \\
& = 2 \Biggl|\ee \Biggl[\int_{0}^{T} (X_s-X_{\varphi_s^{\eta}}) \bar{a}(s,S_s) \text{d}s \Biggr] \Biggr| \\
& \leq 2 \eta \left( T \int_{0}^{T} \ee[a^2(s,S_s)]  \text{d}s \right)^{\frac{1}{2}} ,
\end{align*}
where we used the Cauchy-Schwartz inequality in the last step. Due to H1 and H2 there is a bounded constant $C$ such that
\begin{align*}
\bar{a}^2(t,S_t) & = \left(\sigma(t,S_t) S_t \frac{\partial \sigma}{\partial s}(t,S_t)+\sigma^2(t,S_s)) \right)^2 S_s^2 \left( \frac{\partial^2 u}{\partial s^2}(t,S_t) \right)^2  \\
& \leq C S_s^2 \left( \frac{\partial^2 u}{\partial s^2}(t,S_t) \right)^2.
\end{align*}
By the last estimate in Lemma \ref{lem:bnds}
\begin{align*}
\ee \left[ \int_0^T S^2_t \left( \frac{\partial^2 u}{\partial s^2}(t,S_t) \right)^2  \dd t \right] < \infty,
\end{align*}
which implies that $\int_{0}^{T} \ee[\bar{a}^2(s,S_s)]  \text{d}s < \infty$. From the calculations above it follows that the difference between $\eta^2 \ee[N_T^\eta]$ and $\ee[X]_T$ goes to $0$ as $\eta$ approaches $0$, and thus
\begin{align*}
\lim_{\eta \downarrow 0} \eta^2 \ee[N_T^\eta]=\ee[X]_T=\ee \left[ \int_0^T S^2_t \sigma^2(t,S_t) \left( \frac{\partial^2 u}{\partial s^2}(t,S_t) \right)^2  \dd t \right]<\infty
\end{align*}
as was to be shown. 
\end{proof}

From Theorem \ref{prop:main_1} and \ref{prop:main_2} the following corollary regarding the rate of convergence with respect to $H^\eta_T$ follows.

\begin{corollary} \label{prop:main_3}
Let the payoff function be defined by $\Phi(x)=(x-K)^+$, let $\varphi_t^{\eta}$ be given by Definition \ref{Def:hedgehittimes}, and let $\tilde{\sigma}$ satisfy Assumption \ref{Assump:sigma}, then
\begin{multline}
\lim_{\eta \downarrow 0} H_\eta(T) \ee [\mathcal{R}_a^2(\eta)]=\frac{1}{6} \int_0^T \ee \left[ e^{-2 r t} S^2_t \sigma^2(t,S_t) \right] \, \dd t \\
\times \int_0^T \ee \left[ S^2_t \sigma^2(t,S_t) \left( \frac{\partial^2 u}{\partial s^2}(t,S_t) \right)^2  \right]   \, \dd t  \, .
\label{eq:roc_2}
\end{multline}
\end{corollary}

\begin{remark}
Note that we do not gain anything in terms of rate of convergence applying the adaptive strategy compared to the strategy using a equidistant time net. Also note the similarity between expression (\ref{eq:conv_eqdist}) and (\ref{eq:roc_2}). It is hard to say anything general about the ordering between (\ref{eq:conv_eqdist}) and (\ref{eq:roc_2}), i.e. when adaptive rebalancing is more feasible than equidistant rebalancing or vice versa.
\end{remark}

\section{Discussion} \label{sec:discussion}

The limiting distribution of the normalized difference in hedge ratios between the hedge portfolio and the derivative has been derived. Using this result the limiting expected squared hedging error as the number of rebalancings approach infinity is investigated. We show that $H_\eta(T) \ee [\mathcal{R}_a^2(\eta)]$ converges to a nondegenerate limit as $\eta$ approaches $0$, and derive an explicit expression of this limit. This also shows that we do not gain anything in terms in order of convergence when using the adaptive rebalancing strategy compared to when the hedge portfolio is rebalanced on an equidistant time grid.

Further research could be directed towards more general payoff functions or other market models. One direction could be to study multi dimensional markets that can be made complete by introducing more hedge instruments, e.g. the Heston model.

\appendix
\section{A technical lemma}

\begin{lemma} \label{lem:bnds}
Let the payoff function be defined by $\Phi(x)=(x-K)^+$, let
$\varphi_t^{\eta}$ be given by Definition \ref{Def:hedgehittimes}, let $u$ be
given by (\ref{eqn:price_u}), and let $\tilde{\sigma}$ satisfy Assumption~\ref{Assump:sigma}, then
\begin{enumerate}
\item  for each $t \in [0,T)$ the function $(\partial / \partial s) u(t,s)$ is continuous and strictly increasing in $s$, $\lim_{s \downarrow 0} (\partial / \partial s)u(t,s)=0$ and $\lim_{s \uparrow \infty} (\partial / \partial s)u(t,s)=1$,
\item for $\epsilon > 0$ and $|k| < \infty$ there is a bounded positive constant $C$ such that for all $(t,s) \in [0,T-\epsilon] \times \rr_+$
\begin{align}
C^{-1} e^{-C\log^2(s)} \leq s^k \frac{\partial^2 u}{\partial s^2} (t,s) & \leq C e^{-C^{-1}\log^2(s)} , \label{lem:bnd1}
\intertext{and}
\left| \frac{\partial^3 u}{\partial s^3}(t,s) \right| & \leq C , \label{lem:bnd2}
\end{align}
\item there is a bounded positive constant $C$ such that
\begin{align}
\int_0^T \ee\left[ S_t^2 \left(\frac{\partial^2 u}{\partial s^2}(t,S_t)\right)^2 \right] \text{d} t & \leq C . \label{lem:bnd3}
\end{align}
\end{enumerate}
\end{lemma}

\begin{proof} Throughout the proof we will let $C$ denote a positive bounded constant whose value may change from line to line.

\textit{(i)}: Using \eqref{eqn:d2uds2}--\eqref{eqn:sqineq} it is seen that for $t \in [0,T)$ the derivative of $(\partial/\partial s)u(t,s)$ with respect to $s$ is bounded and alway positive which implies that $(\partial/\partial s)u(t,s)$ is continuous and strictly increasing in $s$.

Let the process $\hat{S}$ be defined by the dynamics
\begin{align*}
\text{d} \hat{S}_t = -r\hat{S}_t \text{d} t + \sigma(T-t,\hat{S}_t) \hat{S}_t \text{d} \hat{W}_t,
\end{align*}
where $\hat{W}$ is a standard Wiener process, and let the process $\hat{Y}$ be defined by $\hat{Y}_t=\log(\hat{S}_t)$.
Due to the put-call duality the call price may be represented by \citep[see][Corollary 3.2, which holds under H1-H3]{Jourdain:2007}
\begin{align}
u(t,s) = e^{-r(T-t)} \mathbb{E}_{t,s}[(S_T-K)^+] = \mathbb{E}_{t,K}[(s-\hat{S}_T)^+]. \label{eqn:u}
\end{align}
Let $p_{\hat{S}}$ denote the transition density of the process $\hat{S}$ (i.e. $p_{\hat{S}}(t,s,T,A)=\mathbb{P}_{t,s}(\hat{S}_T \in A)$) and let $p_{\hat{Y}}$ denote the transition density of the process $\hat{Y}$.  According to \citet[Theorem 6.5.4]{Friedman:2006}, under H1, the transition density $p_{\hat{Y}}$ exists and is given by the fundamental solution to the following PDE
\begin{multline*}
\frac{\partial}{\partial t}p_{\hat{Y}}(t,y,t',y')+ \left( -r-\frac{1}{2} \tilde{\sigma}^2(T-t,y) \right) \frac{\partial}{\partial y}p_{\hat{Y}}(t,y,t',y') \\
+ \frac{1}{2} \tilde{\sigma}^2(T-t,y) \frac{\partial^2}{\partial y^2}p_{\hat{Y}}(t,y,t',y') = 0.
\end{multline*}
Differentiating \eqref{eqn:u} once with respect to $s$ yields
\begin{align*}
\frac{\partial u}{\partial s} (t,s) = \int_0^s p_{\hat{S}}(t,K,T,x) \text{d}x = \int_{-\infty}^{\log(s)} p_{\hat{Y}}(t,\log(K),T,y) \text{d}y.
\end{align*}
According to \citet[][Theorem 6.4.5]{Friedman:2006} there exists a bounded constant $C$ such that $p_{\hat{Y}}(t,y,t',y')  \leq C (t'-t)^{-\frac{1}{2}} \exp\{-(y'-y)^2 (C(t'-t))^{-1}\}$. Using this we get that
\begin{multline*}
0 \leq \ee_{t,\log(K)} \left[I({\hat{Y}_T<\log(s)}) \right] \\
\leq \int_{-\infty}^{\log(s)} \frac{C}{\sqrt{T-t}} e^{-\frac{(y-\log(K))^2}{C(T-t)}} \dd y = \int_{-\infty}^{\frac{\log(s/K)}{\sqrt{T-t}}} C e^{-\frac{y^2}{C}} \dd y,
\end{multline*}
which goes to $0$ as $s$ approaches $0$, and which proves that $\lim_{s \downarrow 0} u_s(t,s)=0$. To show the other limit we note that
\begin{multline*}
1 \geq \ee_{t,\log(K)} \left[I({\hat{Y}_T<\log(s)}) \right]=1-\ee_{t,\log(K)} \left[I({\hat{Y}_T \geq \log(s)}) \right] \\
\geq 1-\int_{\log(s)}^{\infty} \frac{C}{\sqrt{T-t}} e^{-\frac{(y-\log(K))^2}{C(T-t)}} \dd y=1-\int_{\frac{\log(s/K)}{\sqrt{T-t}}}^{\infty} C e^{-\frac{y^2}{C}} \dd y,
\end{multline*}
where the last term goes to zero as $s$ approaches infinity, and which proves that $\lim_{s \uparrow \infty} u_s(t,s)=1$

\textit{(ii) Eqn. \eqref{lem:bnd1}}: Differentiating \eqref{eqn:u} twice with respect to $s$ we get
\begin{align}
\frac{\partial^2 u}{\partial s^2} (t,s) = p_{\hat{S}}(t,K,T,s) =  \frac{p_{\hat{Y}}(t,\log(K),T,\log(s))}{s}. \label{eqn:d2uds2}
\end{align}
According to \citet[][Theorem 1]{Aronson:1967}, which holds under H1, $p_{\hat{Y}}$ may be bounded as
\begin{align}
C^{-1} \frac{e^{-C\frac{(y'-y)^2}{t'-t}}}{\sqrt{t'-t}} \leq p_{\hat{Y}}(t,y,t',y') \leq C \frac{e^{-C^{-1} \frac{(y'-y)^2}{t'-t}}}{\sqrt{t'-t}}. \label{eqn:ghatbnd}
\end{align}
for some bounded positive constant $C$. Using the identity \eqref{eqn:d2uds2} and the inequalities \eqref{eqn:ghatbnd} together with the inequalities 
\begin{align}
-2x^2-c^2/4 \leq -x^2+cx \leq -x^2/2+c^2/2, \label{eqn:sqineq}
\end{align}
which holds for all $x,c \in \mathbb{R}$, we have that there is a bounded positive constant, $C$, such that
\begin{align*}
C^{-1} e^{-C \log^2(s)} \leq s^k \frac{\partial^2 u}{\partial s^2} (t,s) & \leq C e^{-C^{-1} \log^2(s)} ,
\end{align*}
which holds for all $(t,s) \in [\epsilon,T-\epsilon] \times \mathbb{R}_+$ and all finite $k$.

\textit{(ii) Eqn. \eqref{lem:bnd2}}: Differentiating \eqref{eqn:u} three times we get that
\begin{align}
\left| \frac{\partial^3 u}{\partial s^3} (t,s) \right| \leq \left|
\frac{p_{\hat{Y}}(t,\log(K),T,\log(s))}{s^2}  \right| + \left| \frac{1}{s^2}
\left.\frac{\partial}{\partial y} p_{\hat{Y}}(t,\log(K),T,y)
\right|_{y=\log(s)} \right|. \label{eqn:pfbnd2}
\end{align}
The first term on the right hand side in \eqref{eqn:pfbnd2} may be bounded using \eqref{eqn:ghatbnd} and \eqref{eqn:sqineq}, and it is seen that this term is bounded for all $(t,s) \in [\epsilon,T-\epsilon] \times \mathbb{R_+}$. Using \citet[][Theorem 6.4.7]{Friedman:2006}, which holds under H1 and H2, and \citet[][Theorem 6.4.5]{Friedman:2006}, which holds under H1, we have that there is a bounded positive constant $C$ such that
\begin{align*}
\left| \frac{\partial}{\partial y'} p_{\hat{Y}}(t,y,t',y')
\right| \leq \frac{C}{(t'-t)} e^{-\frac{(y'-y)^2}{C(t'-t)}},
\end{align*}
for all $(t,s) \in [0,T] \times \mathbb{R_+}$ (alternatively see \citet[][Theorem 9.6.7]{Friedman:2008}). The above inequality together with \eqref{eqn:sqineq} yields that also the second term on the right hand side in \eqref{eqn:pfbnd2} is bounded for all $(t,s) \in [\epsilon,T-\epsilon] \times \mathbb{R_+}$. Thus, for some bounded constant $C$, $|(\partial^3/\partial s^3)u(t,s)| \leq C$ for all $(t,s) \in [\epsilon,T-\epsilon] \times \mathbb{R}_+$.

\textit{(iii) Eqn. \eqref{lem:bnd3}}: Recall that $Y_t = \log(S_t)$. According to \citet[][Theorem 6.4.5]{Friedman:2006} there exists a bounded constant $C$ such that $p_Y(t,y,t',y')  \leq C (t'-t)^{-\frac{1}{2}} \exp\{-(y'-y)^2(C(t'-t))^{-1}\}$. Using this together with the identity \eqref{eqn:d2uds2} and the inequality to the right in \eqref{eqn:ghatbnd} we get that
\begin{align*}
\ee \left[s^2 \left(\frac{\partial^2 u}{\partial s^2} (t,s)\right)^2  \right] \leq C \int_{\mathbb{R}} \frac{e^{-\frac{(y-\log(K))^2}{C(T-t)}}}{T-t} \frac{e^{- \frac{(y-\log(s_0))^2}{Ct}}}{\sqrt{t}} \text{d} y = \frac{\sqrt{C \pi} e^{-\frac{\log^2(s_0/K)}{CT}}}{\sqrt{T(T-t)}},
\end{align*}
for some bounded positive constant $C$. The above expression is integrable with respect to $t$ over $[0,T]$, which proves~\eqref{lem:bnd3}.
\end{proof}

\end{document}